\documentclass[two-side,11pt]{article}
\pdfoutput=1
\usepackage{alex}
\usepackage{algorithm}
\usepackage{algorithmic}
\usepackage{graphicx}
\usepackage{geometry}
\usepackage{natbib}
\newcommand{\myfig}[1]{\includegraphics[height=0.4\textheight]{#1}}

\graphicspath{{arxiv/}}

\newcommand{\emp}{\mathrm{emp}}

\begin{document}

\title{Direct Optimization of Ranking Measures}
\author{%
  Quoc V. Le\footnote{{\tt Quoc.Le@tuebingen.mpg.de}, 
    Max Planck Institute for Biological Cybernetics, 
    Spemannstr.~38, 72076 T\"ubingen, Germany}
  ~and~
  Alex J. Smola\footnote{{\tt Alex.Smola@nicta.com.au},
    Statistical Machine Learning Program, 
    NICTA and ANU, Canberra, 2600 ACT, Australia}
} 

\maketitle

\begin{abstract}%
  Web page ranking and collaborative filtering require the optimization
  of sophisticated performance measures. Current Support Vector
  approaches are unable to optimize them directly and focus on pairwise
  comparisons instead. We present a new approach which allows
  \emph{direct} optimization of the relevant loss functions.
  
  This is achieved via structured estimation in Hilbert spaces. It is
  most related to Max-Margin-Markov networks optimization of
  multivariate performance measures. Key to our approach is that during
  training the ranking problem can be viewed as a linear assignment
  problem, which can be solved by the Hungarian Marriage algorithm. At
  test time, a sort operation is sufficient, as our algorithm assigns a
  relevance score to every (document, query) pair. Experiments show that
  the our algorithm is fast and that it works very well.
\end{abstract}

 
\section{Introduction}

Ranking, and web-page ranking in particular, has long been a fertile
research topic of machine learning. It is now commonly accepted that
ranking can be treated as a supervised learning problem, leading to
better performance than using one feature alone
\citep{BurShaRenLazDeeHamHul05, CaoXuLiuLiHuaHon06}. Learning to rank can
be viewed as an attempt of learning an ordering of the data (e.g.~web
pages).  Although ideally one might like to have a ranker that learns
the partial ordering of \emph{all} the matching web pages, users are
most concerned with the topmost (part of the) results returned by the
system. See for instance \citep{CaoXuLiuLiHuaHon06} for a discussion.
 
This is manifest in corresponding performance measures developed in 
information retrieval, such as \emph{Normalized Discounted
  Cumulative Gain (NDCG)}, \emph{Mean Reciprocal Rank (MRR)},
\emph{Precision@n}, or ’\emph{Expected Rank Utility (ERU)}. They are
used to address the issue of evaluating rankers, search
engines or recommender sytems \citep{Voorhees01, JarKek00, BreHecKar98,
  BasHof04}.

Ranking methods have come a long way in the past years. Beginning with
vector space models \citep{Salton71, SalMcG83}, various feature based
methods have been proposed \citep{LeeChuSea97}.  Popular set of features
include BM25 \citep{Robertsonetal94} or its variants \citep{RobHul00}.
Following the intent of \cite{RicPraBri06} we show that when combining
such methods with machine learning, the performance of the ranker can be
increased significantly.

Over the past decade, many machine learning methods have been proposed.
Ordinal regression \citep{HerGraObe00, ChuKee05} using a SVM-like large
margin method and Neural Networks
\citep{BurShaRenLazDeeHamHul05} were some of the first
approaches. This was followed by Perceptrons \citep{CraSin02}, and
online methods, such as \citep{CraSin05, BasHof04}. The state of the art
is essentially to describe the partial order by a directed acyclic graph
and to associate the cost incurred by ranking incorrectly with the edges
of this graph.  These methods aim at finding the best ordering function
over the returned documents. However, it is difficult to express complex
(yet commonly used) measures in this framework.

Only recently two theoretical papers \citep{Rudin06,CosZha06} discuss the
issue of learning ranking with preference to top scoring documents.
However, the cost function of \citep{Rudin06} is only vaguely related to
the cost function used in evaluating the performance of the ranker.
\citep{CosZha06} argue that, in the limit of large samples,
regression on the labels may be sufficient.

Our work uses the framework of Support Vector Machines for Structured
Outputs \citep{TsoJoaHofAlt05,Joachims05,TasGueKol03} to deal with the
inherent non-convexity of the performance measures in ranking. Due to
the capacity control inherent in kernel methods it generalizes well to
test observations \citep{SchSmo02}. The optimization problem we propose
is very general: it covers a broad range of existing criteria in a
plug-and-play manner. It extends to position-dependent ranking and
diversity-based scores.

Of particular relevance are two recent papers \citep{Joachims05, BurLeRag07} to
address the complication of the information retrieval loss functions.
More specifically, \citep{Joachims05} shows that two ranking-related
scores, Precision@n and the Area under the ROC curve, can
be optimized by using a variant of Support Vector Machines for
Structured Outputs (SVMStruct). We use a similar strategy in
our algorithm to obtain a Direct Optimization of Ranking Measures (DORM)
using the inequalities proposed in \citep{TsoJoaHofAlt05}. \citep{BurLeRag07}
considered a similar problem without the convex relaxation and instead
they optimize the nonconvex cost functions directly by only dealing with
their {\em gradients}.

{\bfseries Outline:} After a summary of
structured estimation we discuss performance measures in
information retrieval (Section~3) and we express them as inner
products. In Section~\ref{sec:General} we compute a convex relaxation of the
performance criterion and show how it can be solved efficiently using
linear assignment. Experiments on web search and
collaborative filtering show that DORM is fast and works well.

\section{Structured Estimation}
\label{sec:svmstruct}

In the following we will develop a method to rank objects (e.g.\ 
documents $d$) subject to some query $q$ by means of some function
$g(d,q)$. Obviously we want to ensure that highly relevant documents
will have a high score, i.e.\ a large value of $g$. At the same time, we
want to ensure that the ranking obtained is optimal with respect to the
relevant ranking score. For instance for NDCG@10, i.e.\ a score where
only the first 10 retrieved documents matter, it is not very important
what particular values of a score $g$ will assign to highly irrelevant
pages, provided that they remain below the acceptance threshold.

Obviously, we could use \emph{engineering} skills to construct a
reasonable function (PageRank is an example of such a function).
However, we can also use statistics and machine learning to guide us
find a function that is optimized for this purpose. This leads to a more
optimized way of finding such a function, removing the need for educated
guesses. The particular tool we use is max margin structured estimation,
as described in \cite{TsoJoaHofAlt05}. See the original reference for a
more detailed discussion.

\subsection{Problem Setting}
\label{sec:probset}

Large margin structured estimation, as proposed by
\citep{TasGueKol03,TsoJoaHofAlt05},
is a general strategy to solve estimation problems of mapping $\Xcal \to
\Zcal$ by finding related optimization problems. More concretely it
solves the estimation problem of finding a matching $z \in \Zcal$ from a
set of (structured) estimates, given patterns
$x \in \Xcal$, by finding a function $f(x,z)$ such that
\begin{align}
  z^*(x) := \argmax_{z \in \Zcal} f(x,z).
\end{align}
This means that instead of finding a mapping $\Xcal \to \Zcal$ directly,
the problem is reduced to finding a real valued function on $\Xcal
\times \Zcal$. 

In the ranking case, $x \in \Xcal$ corresponds to a set
of documents with a corresponding query, whereas $z$ would correspond to
the permutation which orders the documents such that the most relevant
ones are ranked first. Consequently $f$ will be a function of the
documents, query, and a permutation. It is then our goal to find such an
$f$ that it is maximized for the ``correct'' permutation. 

To assess how well the estimate $z^*(x)$ performs we need to introduce a
loss function $\Delta(y, z)$, depending on $z$ and some reference
labels $y$, which determines the loss at $z$. For instance, if we
want to solve the binary classification problem $y, z \in \cbr{\pm
  1}$, where $y$ is be the observed label and $z$ the
estimate we could choose $\Delta(y, z) = 1 - \delta_{y, z}$. That is, we
incur a penalty of $1$ if we make a mistake, and no penalty if we
estimate $y = z$. In the regression case, this could be $\Delta(y, z) =
(y - z)^2$. Finally, in the sequence annotation case, where both $y$ and
$z$ are binary sequences, \citep{TasGueKol03,TsoJoaHofAlt05} use the
Hamming loss. In the ranking case, which we will discuss in
Section~\ref{sec:Measures}, the loss $\Delta$ will correspond to the
\emph{relative} regret incurred by ranking documents in a suboptimal
fashion with respect to WTA, MRR, DCG, NDCG, ERU or a similar
criterion. Moreover $y$ will correspond to the relevance scores assigned
to various documents by reference users. 

In summary, it is our goal to find some function $f$ to minimize the
error incurred by $f$ on a set of observations $X = \cbr{x_1, \ldots,
  x_m}$ and reference labels $Y = \cbr{y_1, \ldots, y_m}$
\begin{align}
  \label{eq:remp}
  R_\emp[f, X, Y] := \sum_{i=1}^m \Delta(y_i, \argmax_{z \in
  \Zcal} f(x_i, z)).
\end{align}
We will refer to $R_\emp[f, X, Y]$ as the empirical risk. Direct
minimization of the latter with respect to $f$ is difficult:
\begin{itemize}
\item It is a highly nonconvex optimization problem. This makes
  practical optimization extremely difficult, as the problem has many
  local minima.  
\item Good performance with respect to the empirical risk $R_\emp[f, X,
  Y]$ does \emph{not} result in good performance on an unseen test set. In
  practice, strict minimization of the empirical risk virtually ensures
  bad performance on a test set due to overfitting. This issue has been
  discussed extensively in the machine learning literature (see
  e.g.~\citep{Vapnik82}). 
\end{itemize}
To deal with the second problem we will add a regularization term (here
a quadratic penalty) to the empirical risk. To address the first problem,
we will compute a convex upper bound on the loss $\Delta(y_i, \argmax_{z \in
  \Zcal} f(x_i, z))$. 

\subsection{Convex Upper Bound}

The key inequality we exploit in obtaining a convex upper bound on the
problem of minimizing the loss $\Delta(y, z)$ is the following lemma,
which is essentially due to \citep{TsoJoaHofAlt05}. 
\begin{lemma}
  \label{lem:relax}
  Let $f: \Xcal \times \Zcal \to \RR$ and $\Delta: \Ycal \times \Zcal
  \to \RR$ and let $z_0 \in \Zcal$. Moreover let $\xi \in \RR$. In this
  case, if $\xi$ satisfies  
  \begin{align*}
    f(x,z_0) - f(x,z) \geq \Delta(y,z) - \Delta(y,z_0) - \xi \text{ for all } z
    \in \Xcal,
  \end{align*}
  then $\xi \geq \Delta(y,\argmax_{z \in \Zcal} f(x,z)) - \Delta(y,
  z_0)$. Moreover, the constraints on $\xi$ and $f$ are linear. 
\end{lemma}
\begin{proof}
  Linearity is obvious, as $\xi$ and $f$ only appear
  as individual terms. To see the first claim, denote by $z^*(x) :=
  \argmax_{z \in \Zcal} f(x,z)$. Since the inequality needs to hold for
  all $z$, it holds in particular for $z^*(x)$. This implies
  \begin{align*}
    0 \geq f(x,z_0) - f(x,z^*(x)) \geq \Delta(y,z^*(x)) - \Delta(y,z_0) - \xi.
  \end{align*}
  The first inequality holds by construction of $z^*(x)$. Rearrangement
  proves the claim. 
\end{proof}
Typically one chooses $z_0$ to be the minimizer of $\Delta$ and one
assumes that the loss for $z_0$ vanishes. In this case $\xi \geq
\Delta(z^*)$. Note that this convex upper bound is tight for $z_0 = z^*$
and if the minimal $\xi$ satisfying this inequality is chosen. 

\subsection{Kernels}

The last ingredient to obtain a convex optimization problem is a
suitable function class for $f$. In principle, any class, such as
Decision Trees, Neural Networks, convex combinations of weak learners as
they occur in Boosting, etc.\ would be acceptable. For convenience we
choose $f$ via
\begin{align}
  \label{eq:fexpand}
  f(x,z) = \inner{\Phi(x,z)}{w}.
\end{align}
Here $\Phi(x,z)$ is a feature map and $w$ is a corresponding weight
vector. The advantage of this formulation is that by choosing different
maps $\Phi$ it becomes possible to incorporate prior knowledge
efficiently. Moreover, it is possible to express the arising
optimization and estimation problem in terms of the kernel functions
\begin{align}
  k((x,z), (x',z')) := \inner{\Phi(x,z)}{\Phi(x',z')}
\end{align}
without the need to evaluate $\Phi(x,z)$ explicitly. This allows us to
work with infinite-dimensional feature spaces while keeping the number
parameters of the optimization problem finite. Moreover, we can control
the complexity of $f$ by keeping $\nbr{w}$ sufficiently small (for
binary classification this amounts to \emph{large margin}
classification). 

Denote by $z_i$ the value of $z$ for which $\Delta(y_i, z)$ is
minimized. Moreover, let $C > 0$ be a regularization constant specifying
the trade-off between empirical risk minimization and the quest for a
``simple'' function $f$.  Combining \eq{eq:remp}, Lemma~\ref{lem:relax},
and \eq{eq:fexpand} one arrives at the following optimization problem:
\begin{subequations}
  \label{eq:primal}
  \begin{align}
    \label{eq:primal-obj}
    \mini_{w,\xi} ~  
    &\frac{1}{2} \nbr{w}^2 + C\sum_{i=1}^m \xi_i \\
    \label{eq:primal-cons}
    \text{ subject to } &
    \inner{w}{\Phi(x_i, z_i) - \Phi(x_i, z)} \geq \Delta(y_i, z) - \xi_i \\
  \nonumber
  \text{for all }
  &\xi_i \geq 0 \text{ and } z \in \Zcal \text{ and } i \in \cbr{1, \ldots, m}.
\end{align}
\end{subequations}
In the ranking case we will assume (without
loss of generality) that the documents are already ordered in decreasing
order of relevance. In this case $z_i$ will correspond to the unit
permutation which leaves the order of documents unchanged.

\subsection{Optimization}

One may show \citep{TasGueKol03} that the solution of \eq{eq:primal} is
given by 
\begin{align}
  f(x',z') = \sum_{i, z} \alpha_{iz} k((x_i, z),(x', z')).
\end{align}
This fact is also commonly referred to as the Representer Theorem
\citep{SchSmo02}. The coefficients $\alpha_{iz}$ are obtained by solving
the dual optimization problem of \eq{eq:primal}:
\begin{subequations}
  \label{eq:dual}
  \begin{align}
    \mini_{\alpha} ~ & \frac{1}{2} \sum_{i,j, z, z'} \alpha_{iz} \alpha_{jz'}
    k((x_i, z),(x_j, z'))
    \label{eq:dual-obj}
    -\sum_{i, z} \Delta(y_i, z) \alpha_{iz} \\
    \label{eq:dual-cons}
    \text{subject to} ~ & \sum_{z} \alpha_{i z} \leq C \text{
      and } \alpha_{i z} \geq 0 \text{ for all } i \text{ and } z.
\end{align}
\end{subequations}
Solving the optimization problem \eq{eq:dual} presents a formidable
challenge. In particular, for large $\Zcal$ (e.g.\ the space of all
permutations over a set of documents) the number of variables is
prohibitively large and it is essentially impossible to find an optimal
solution within a practical amount of time. Instead, one may use column
generation \citep{TsoJoaHofAlt05} to find an approximate solution in
polynomial time. The key idea in this is to check the constraints
\eq{eq:primal-cons} to find out which of them are violated for the
current set of parameters and to use this information to improve the
value of the optimization problem. That is, one needs to find
\begin{align}
  \label{eq:colgen}
    \argmax_{z \in \Zcal} \Delta(y_i, z) + \inner{w}{\Phi(x_i, z)},
\end{align}
as this is the term for which the constraint \eq{eq:primal-cons} becomes
tightest. If $\Zcal$ is a small finite set of values, this is best
achieved by brute force evaluation. For binary sequences, one often uses
dynamic programming. In the case of ranking, where $\Zcal$ is the space of
permutations, we shall see that \eq{eq:colgen}
can be cast as a linear assignment problem. 

\begin{algorithm}[tb]
  \caption{Column Generation}
  \label{alg:colgen}
\begin{algorithmic}
  \STATE {\bfseries Input:} data $x_i$, labels $y_i$,
  sample size $m$, tolerance $\epsilon$ 
  \STATE Initialize $S_i = \emptyset$ for all $i$, and $w = 0$.
  \REPEAT
  \FOR{$i=1$ {\bfseries to} $m$} 
  \STATE $w = \sum_{i} \sum_{z \in S_i} \alpha_{iz} \Phi(x_i, z)$
  \STATE $z^* = \argmax_{z \in \Zcal} \inner{w}{\Phi(x_i, z)} +
  \Delta(y_i, z)$
  \STATE $\xi = \max(0, \max_{z \in S_i} \inner{w}{\Phi(x_i, z)}
    + \Delta(y_i, z))$
  \IF{$\inner{w}{\Phi(x_i, z^*)} + \Delta(y_i, z) > \xi + \epsilon$}
  \STATE Increase constraint set $S_i \leftarrow S_i \cup {z^*}$
  \STATE Optimize \eq{eq:dual} using only $\alpha_{iz}$ where $z \in S_i$.
  \ENDIF
  \ENDFOR
  \UNTIL{$S$ has not changed in this iteration}
\end{algorithmic}
\end{algorithm}

The Algorithm \ref{alg:colgen} has good convergence properties. It
follows from \citep{TsoJoaHofAlt05} that it terminates after adding at
most $\max\sbr{2n\bar{\Delta}/\epsilon, 8C \bar\Delta R^2/\epsilon^2}$
steps, where $\bar\Delta$ is an upper bound on the loss $\Delta(y_i, z)$
and $R$ is an upper bound on $\nbr{\Phi(x_i, z)}$.

To adapt the above framework to the ranking setting, we need to address
three issues: a) we need to derive a loss function $\Delta$ for ranking, b) we
need to develop a suitable feature map $\Phi$ which takes document
collections, queries, and permutations into account, and c) we need to
find an algorithm to solve \eq{eq:colgen} efficiently. 

\section{Ranking and Loss Functions}
\label{sec:Measures}

\subsection{A Formal Description}
\label{sec:formaldef}

For efficiency, commercial rankers, search engines, or recommender
systems, usually employ a document-at-a-time approach to answer a query
$q$: a list of candidate documents is evaluated (while retaining a heap
of the $n$ top scoring documents) by evaluating the relevance for a
(document, query)-pari one at a time. For this purpose a score function
$g(d,q)$ is needed, which assigns a score to every document given the
query.\footnote{For computational efficiency (not for theoretical
  reasons) it is \emph{not} desirable that $f$ depends on $\cbr{d_1,
    \ldots, d_l}$ jointly.}  Performance of the ranker is typically
measured by means of a set of labels $y := \cbr{r_1, \ldots, r_l}$ with
$r_i \in [0, \ldots, r_{\max}]$, where $0$ corresponds to 'irrelevant'
and $r_{\max}$ corresponds to 'highly relevant'. Training instances
contain document query pairs that are labelled by experts.  
Such data for commercial search engines or
recommender systems often have less than ten levels of relevance.

At training time we are given $m$ instances of queries $q_i$, document
collections $D_i$ of cardinality $l_i$ and labels $y_i$ with $|y_i| =
|D_i|$. In the context of the previous section a set of documents $D_i$
in combination with a matching query $q_i$ will play the role of a
pattern, i.e.\ $x_i := (q_i, d_{iq}, \ldots, d_{il_i})$. Likewise the
reference labels $r_{ij} \in y_i$ consist of the corresponding expert
ratings for the documents $d_{ij}$.

We want to find some mapping $f$, such that the ordering of a new
collection of documents $d_{1}, \ldots, d_{l}$ obtained by sorting
$g(d_{i}, q)$ agrees well with $y$ in expectation. We
would like to obtain a \emph{single} ranking which will perform well on
the query for a given performance measure, unlike \citep{Matveevaetal06} who
use a cascade of rankers. 

Note that there is also a processing step associated with ranking
documents: for each document query pair $(d,q)$, we have to construct
a feature vector $x$. In this paper, we assume that the feature
vector is given, and also use $(d,q)$ to mean $x$. For instance BM25
\citep{Robertsonetal94, RobHul00}, date, click-through logs \citep{Joachims02b} have proved to be
an effective set of features.

Many widely-used performance measures in the information retrieval
community are irreducibly multivariate and permutation based. By
permutation-based, we mean that the performance measure can be computed
by comparing the two sets of ordering. For example, 'Winner Takes All'
(WTA), Mean Reciprocal Rank (MRR) \citep{Voorhees01}, Mean Average Precision
(MAP), and Normalized Discounted Cumulative Gain (NDCG) \citep{JarKek00}
all fullfil this property.

It is our goal to find a suitable permutation $\pi(D,q,g)$ obtained for
the collection $D$ of documents $d_i$ given the query $q$ and the
scoring function $g$. We will drop the arguments $D, q, g$, wherever it
is clear from the context. Moreover, given a vector $v \in \RR^m$ we
denote by $v(\pi)$ the permutation of $v$ according to $\pi$, i.e.\ 
$v(\pi)_i = v_{\pi(i)}$. Finally, without loss of generality (and for
notational convenience) we assume
that $y$ is sorted in \emph{descending} order, i.e.\ most relevant
documents first. That is, the identical permutation $\pi = \one$ will
correspond to the sorting which returns the most relevant documents
first with respect to the reference labeling. 

Note that $\pi$ will play the role of $z$ of
Section~\ref{sec:svmstruct}. Likewise we will denote by $\Pi$ the space
of all permutations (i.e.\ $\Zcal = \Pi$). 

\begin{table}[tb]
  \centering
  \begin{tabular}{|l||l|} \hline
    Variable & Meaning \\ \hline
    $x = (q, D)$ & document-query pair\\ 
    \hline
    $q_i$ & $i$-th query \\ \hline
    $l_i$ & number of documents for $q_i$\\\hline
    $D_i = \cbr{d_{i1}, \ldots, d_{il_i}}$ & documents for $q_i$\\ \hline
    $r_{ij} \in [0, \ldots, r_{\max}]$ & relevance of document $d_{ij}$
    \\\hline
    $y_i = \cbr{r_{i1}, \ldots, r_{il_i}}$ & reference label \\ \hline
    $f(x, \pi)$ & global scoring function \\\hline
    $g(q_i, d_{ij})$ & individual scoring function \\\hline
    $m$ & number of queries for training\\\hline
  \end{tabular}
\end{table}

\subsection{Scores and Loss}
\label{sec:scoresloss}

\begin{description}
\item[Winner Takes All (WTA):] If the first
  document is relevant, i.e.\ if $y(\pi)_1 = r_1$ the score is 1,
  otherwise 0. We may rewrite this as 
  \begin{align}
    \label{eq:wta}
    \mathrm{WTA}(\pi, y) & = \inner{a(\pi)}{b(y)}\\
    \nonumber
    \text{ where }
    a_i & = \delta(i, 1)
    \text{ and }
    b_i = \delta(r_i, r_1).
  \end{align}
  Note that there may be more than one document which is considered
  relevant. In this case $\mathrm{WTA}(\pi, y)$ will be maximal for
  several classes of permutations. 
\item[Mean Reciprocal Rank (MRR):] We assume that there exists only
  \emph{one} top ranking document. We have
  \begin{align}
    \label{eq:mrr}
    \mathrm{MRR}(\pi, y) & = \inner{a(\pi)}{b(y)}\\
    \nonumber
    \text{ where }
    a_i & = 1/i
    \text{ and }
    b_i = \delta(i, 1).
  \end{align}
  In other words, the reciprocal rank is the inverse of the rank
  assigned to document $d_1$, the most relevant document. MRR
  derives its name from the fact that this quantity is typically
  averaged over several queries.
\item[Discounted Cumulative Gain (DCG and DCG@$n$):] 
  WTA and MRR use only a single entry of $\pi$, namely $\pi(1)$, to
  assess the quality of the ranking. 
  Discounted Cumulative Gains are a more balanced score:
  \begin{align}
    \label{eq:dcg}
    \mathrm{DCG}(\pi, y) & = \inner{a(\pi)}{b(y)}\\
    \nonumber
    \text{ where }
    a_i & = 1/{\log (i+1)}
    \text{ and }
    b_i = 2^{r_i} - 1.
  \end{align}
  Here it pays if a relevant document is retrieved with a high rank, as
  the coefficients $a_i$ are monotonically decreasing. Variants of DCG,
  which do not take all ranks into account, are DCG@$n$. Here 
  $a_i = 1/\log(i+1)$ if $i \leq n$ and $a_i = 0$ otherwise. That is, we
  only care about the $n$ top ranking entries. In search engines the
  \emph{truncation level} $n$ is typically $10$, as this constitutes the
  number of hits returned on the first page of a search.

\item[Normalized Discounted Cumulative Gain (NDCG):]
  A downside of DCG is that its numerical range depends on $y$ (e.g. a collection
  containing many relevant documents will yield a considerably
  larger value at optimality than one containing only irrelevant
  ones). Since $y$ is already sorted it follows that DCG is 
  maximized for the identity permutation $\pi = \one$:
  \begin{align}
    \mathrm{NDCG}(\pi, y) &:= \textstyle \frac{\mathrm{DCG}(\pi,
      y)}{\mathrm{DCG}(\one, y)}\\
    \nonumber
    \text{ and } 
    \mathrm{NDCG}@n(\pi, y) &:= \textstyle \frac{\mathrm{DCG}@n(\pi,
      y)}{\mathrm{NDCG}@n(\one, y)}.
  \end{align}  
  This allows us to define 
  \begin{align}
    \label{eq:ndcg}
    \mathrm{NDCG}(\pi, y) & = \inner{a(\pi)}{b(y)}\\
    \nonumber
    \text{ where } \textstyle
    a_i & = \frac{1}{\log (i+1)}
    \text{ and }
    b_i = \frac{2^{r_i} - 1}{\mathrm{DCG}(\one, y)}.
  \end{align}
  Finally, NDCG@$n$, the measure which this paper focuses on, is
  given by $\inner{a(\pi)}{b(y)}$ where 
  \begin{align}
    \label{eq:ndcg-at-k}
    a_i = 
    \begin{cases}
      \frac{1}{\log (i+1)} & \text{if } i \leq n \\
      0 & \text{else}
    \end{cases}
    \text{and }
    b_i = \textstyle \frac{2^{r_i} - 1}{\mathrm{DCG}@n(\one, y)}.
  \end{align}
\item[Precision@n:] Note that this measure, too, can be expressed by
  $\inner{a(\pi)}{b(y)}$. Here we define
  \begin{align}
    a_i =     
    \begin{cases}
      1/n & \text{if } i \leq n \\
      0 & \text{else}
    \end{cases}
    \text{ and } 
    b_i =     
    \begin{cases}
      1 & \text{if } r_i \text{ correct} \\
      0 & \text{else}
    \end{cases}
    \hspace{-4mm}
  \end{align}
  The main difference to NDCG is that Precision@n has no decay
  factor, weighing the top $n$ answers equally. 

\item[Expected rank utility:] It has an exponential decay in the top
  ranked items and it can be represented as 
\begin{align}
  \mathrm{ERU}(\pi,y) & = \inner{a(\pi)}{b(y)}\\ 
  \nonumber
  \text{ where } a_i & = 
  2^{\frac{1-i}{\alpha - 1}} \text{ and } b_i = \max(r_i - d,0)
\end{align}
Here $d$ is a neutral vote and $\alpha$ is the viewing halflife. The
normalized ERU can also be defined in a similar manner to NDCG. The
(normalized) ERU is often used in collaborative filtering for
recommender systems where the lists of items are often very short.


\end{description}
It is commonly accepted that NDCG@n is a good model of a person's
judgment of a search engine: the results on the first page matter,
between them there should be a decay factor. NDCG has another advantage
that it is more structured and more general than WTA and MRR.  For
collaborative and content filtering, ERU is more popular
\citep{BreHecKar98, BasHof04}.

Since we set out to design a \emph{loss} function as described in
Section~\ref{sec:probset}, we now define the relative loss incurred by
any score of the form $\inner{a(\pi)}{b(y)}$. For convenience we assume
(again) that $\pi = \one$ is the optimal permutation:
\begin{align}
  \label{eq:loss}
  \Delta(y, \pi) := \inner{a(\one)}{b(y)} - \inner{a(\pi)}{b(y)}.
\end{align}

\subsection{Scoring Function}

The final step in our problem setting is to define a suitable function
$f(x, \pi)$ (where $x = (q, D)$) which is maximized for the ``optimal''
permutation. As stated in Section~\ref{sec:formaldef} we require a
function $g(d, q)$ which will assign a relevance score to every
(document, query) pair \emph{independently} at test time. The
Polya-Littlewood-Hardy inequality tells us how we can obtain a suitable
class of functions $f$, given $g$:

\begin{theorem}
  Let $a, b \in \RR^n$ and let $\pi \in \Pi$. Moreover, assume
  that $a$ is sorted in decreasing order. Then $\inner{a}{b(\pi)}$ is
  maximal for the permutation sorting $b$ in decreasing order.
\end{theorem}
Consequently, if we define
\begin{align}
  \label{eq:f-to-g}
  f(x, \pi) = \sum_{i} g(d_i, q) c(\pi)_i
\end{align}
for some decreasing sequence $c$, the maximizer of $f$ will be the one
which sorts the documents in decreasing order of relevance, as assigned
by $g(d_i, q)$. The expansion \eq{eq:f-to-g} also acts as a guidance
when it comes to designing a feature map $\Phi(x, \pi)$, which will map
\emph{all} documents $d_i$, the query $q$, and the permutation $\pi$
jointly into a feature space. More to the point, it will need to reflect
the decomposition into terms related to individual pairs $(d_i, q)$
only.

\subsection{General Position Dependent Loss}

Before we proceed to solving the ranking problem by defining a suitable
feature map, let us briefly consider the most general case we are able
to treat efficiently in our framework.\footnote{More general cases, such
  as a quadratic dependency on positions, while possible, will typically
  lead to optimization problems which cannot be solved in polynomial
  time.} Assuming that we are given a ranking $\pi$ of documents
$\cbr{d_1, \ldots, d_l}$, we define the loss as
\begin{align}
  \Delta(\pi,y) := \sum_{i,j}^l \pi_{ij} C_{ij}(y) + \mathrm{const.}
\end{align}
That is, for every position $j$ we have a cost $C_{ij}(y)$ which is
incurred by placing document $i$ at position $j$. Clearly \eq{eq:loss}
falls into this category: simply choose $C_{ij} = a_i b(y)_j$. For
instance, we might have a web page ranking problem where the first
position should contain a result from a government-related site, the
second page should contain a relevant page from a user-created site,
etc. In other words, this setting would apply to cases where specific
positions in the ranked list are endowed with specific meanings. 

The problem with this procedure is, however, that estimation and
optimization are somewhat more costly than merely sorting a list of
relevance scores: we would want to have a different scoring function for
each position. In other words, the computational cost is dramatically
increased, both in terms of the number of functions needed to compute
the scores of an element  and in terms of the optimization required to
find a permutation which minimizes the loss $\Delta(\pi, y)$. 

\section{Learning Ranking}
\label{sec:General}

\subsection{The Featuremap}

We now expand on the ansatz of \eq{eq:f-to-g}. The linearity of the
inner product $f(x,\pi) = \inner{\Phi(x,\pi)}{w}$, as given by
\eq{eq:fexpand} requires that we should be able to write $\Phi$ as 
\begin{align}
  \label{eq:feature}
  \Phi(x, \pi) = \sum_{i=1}^l c(\pi)_i \phi(d_i, q)
  \text{ where } c \in \RR^m. 
\end{align}
In this case $\inner{w}{\Phi(D, q, \pi)} = \inner{c(\pi)}{p_i}$ where $p_i =
\inner{w}{\phi(d_i, q)}$. The problem of choosing $\phi(d_i, q)$ is
ultimately data dependent, as we do not know in general what data type
the documents and queries are composed of. We will discuss several
choices in the context of the experiments in Section~\ref{sec:experiments}. 

Eq.~\eq{eq:feature} implies that $f$ is given by 
\begin{align}
  f(x, \pi) = \inner{\Phi(x, \pi)}{w} = 
  \sum_{i=1}^l c(\pi)_i \inner{\phi(d_i, q)}{w}.
\end{align}
Hence we can apply the Polya-Littlewood-Hardy inequality and observe
that it is maximized by sorting the terms $\inner{\phi(d_i, q)}{w}$ in
decreasing order just as $c$. Note that this permutation is easily
obtained by applying QuickSort in $O(l \log l)$ time. To
obtain only the $n$ top-ranking terms we can do even better and only
need to expend $O(l + n \log n)$ time, using a QuickSort-style divide
and conquer procedure.

This leaves the issue of choosing $c$. In general we want to choose it
such that a) the margin of \eq{eq:primal} is maximized and that b) the
average length $\nbr{\Phi(x, \pi)}^2$ is small for good generalization
and fast convergence of the implementation.

Since $\Phi$ is linear in $c$, we could employ a kernel optimization
technique to obtain an optimal value of $c$, such as those proposed in
\citep{BouHer02,OngSmoWil02}. While in principle not impossible, this
leads rather inevitably to a highly constrained (at worst semidefinite)
optimization problem in terms of $c$ and $w$. Obtaining an efficient
algorithm for this more general problem is topic of current research. We
report experimental results for different choices of $c$ in
Section~\ref{sec:experiments}. 

The above reasoning is sufficient to apply the optimization problem
described in Section~\ref{sec:svmstruct} to ranking problems. All that
changes with respect to the general case is the choice of loss function
$\Delta$ and the feature map $\Phi(x,\pi)$. In order to obtain an
efficient optimization algorithm we need to overcome one last hurdle: we
need to find an efficient algorithm for finding constraint violators in
\eq{eq:colgen}.

\subsection{Finding Violated Constraints}

Recall \eq{eq:colgen}. In the context of ranking this means that we need
to find the permutation $\pi$ which maximizes
\begin{align}
  \nonumber
  & {\inner{\Phi(x,\pi)}{w}} + {\Delta(y, \pi)} \\
  \label{eq:colgen-perm}
  = & {\sum_{i=1}^l \inner{\phi(d_i, q)}{w} c(\pi)_i}
  + \inner{a(\one)}{b(y)} - \inner{a(\pi)}{b(y)} \\
  \label{eq:colgen-perm-trans}
  = & \inner{c(\pi)}{g} - \inner{a(\pi)}{b(y)} + \mathrm{const.} 
\end{align}
where $g_i = \inner{\phi(d_i, q)}{w}$. Note that
\eq{eq:colgen-perm-trans} is a so-called \emph{linear assignment
  problem} which can be solved by the Hungarian Marriage method: the
Kuhn-Munkres algorithm 
in cubic time. Maximizing
\eq{eq:colgen-perm-trans} amounts to solving 
\begin{align}
  \argmax_{\pi \in \Pi} \sum_{i=1}^m C_{i, \pi_i}
  \text{ where } C_{ij} = c_j g_i - a_j b_i. 
\end{align}
Note that there is a special case in which the problem \eq{eq:colgen}
can be solved by a simple sorting operation: whenever $a = c$ the
problem reduces to maximizing $\inner{a(\pi)}{g - b(y)}$. This choice,
however, is not always desirable as $a$ may be rather degenerate (i.e.\
it may contain many terms with value zero).

\subsection{Solving the Linear Assignment Problem}

It is well known that there exists a convex relaxation of the problem of
maximizing $\sum_i C_{i,\pi_i}$ into a linear problem which leads to an
optimal solution. 
\begin{align}
  \label{eq:kuhnmunkries}
  \maxi_{\pi} ~ & \tr \pi^\top C\\ 
  \nonumber
  \text{ subject to } &  \sum_{i} \pi_{ij} = 1
  \text{ and } 
  \sum_j \pi_{ij} = 1 \text{ and }
  \pi_{ij} \in \cbr{0, 1}
\end{align}
More specifically, the integer linear program can be relaxed by
replacing the integrality constraint $\pi_{ij} \in \cbr{0,1}$ by $\pi_{ij}
\geq 0$ without changing the solution, since the constraint matrix is
totally unimodular. Consequently the vertices of the feasible polytope
are integral, hence also the solution. The dual is
\begin{align}
  \label{eq:kuhnmunkries-dual}
  \mini_{u, v} &\sum_i u_i + v_i
  \text{ subject to } &u_i + v_i \geq C_{ij}. 
\end{align}
The solution of linear assignment problems is a well studied subject.
The original papers by \cite{Kuhn55} and \cite{Munkres57} implied an
algorithm with $O(l^3)$ cost in the number of terms. Later,
\cite{Karp80} suggested an algorithm with expected quadratic time in the
size of the assignment problem (ignoring log-factors). Finally,
\cite{OrlLee93} propose a linear time algorithm for large problems.
Since in our case the number of pages is fairly small (in the order of
50 to 200), we used an existing implementation due to \cite{JonVol87}.
See Section~\ref{sec:web-search} for runtime details. The latter uses
modern techniques for computing the shortest path problem arising in
\eq{eq:kuhnmunkries-dual}.

This means that we can check whether a particular set of documents and
an associated query $(D_i, q_i)$ satisfies the inequality constraints of
the structured estimation problem \eq{eq:primal}. Hence we have the
subroutine necessary to make the algorithm of Section~\ref{sec:svmstruct}
work. In particular, this is the only subroutine we need to replace in
SVMStruct \citep{TsoJoaHofAlt05}. 

\subsection{In a Nutshell}

Before describing the experiments, let us briefly summarize the overall
structure of the algorithm. In completely analogy to \eq{eq:primal} the
primal optimization problem can be stated as
%
%
\begin{subequations}
  \label{eq:primal-rank}
  \begin{align}
    \label{eq:primal-rank-obj}
    \mini_{w,\xi} ~ 
    &\frac{1}{2} \nbr{w}^2 + C\sum_{i=1}^m \xi_i \\
    \nonumber
    \text{ subject to } &
    \inner{w}{\Phi(x_i, \one) - \Phi(x_i, \pi)} \geq
    \Delta(y_i, \pi) - \xi_i \\
    \label{eq:primal-rank-cons}
    \text{for all }
    &\xi_i \geq 0 \text{ and } \pi \in \Pi \text{ and } i \in \cbr{1, \ldots, m}
  \end{align}
\end{subequations}
%
%
The dual problem of \eq{eq:primal} is given by 
\begin{subequations}
  \label{eq:dual-rank}
\begin{align}
  \nonumber
  \mini_{\alpha} ~ & \frac{1}{2} \sum_{i,j, \pi, \pi' \in \Pi}
  \alpha_{i\pi} \alpha_{j\pi'} k((x_i, \pi),(x_j, \pi')) \\
  &- \sum_{i, \pi \in \Pi} \Delta(y_i, \pi) \alpha_{i\pi}
  \label{eq:dual-rank-obj}
  \\
  \text{subject to} ~ & \sum_{\pi \in \Pi} \alpha_{i\pi} \leq C \text{
  and } \alpha_{i\pi} \geq 0 \text{ for all } i \text{ and } \pi.
\end{align}
\end{subequations}
This problem is solved by Algorithm~\ref{alg:dorm}. Finally, documents
on a test set are ranked by $g(d,q) = \inner{w}{\phi(d,q)}$, where 
$w = \sum_{i,\pi} \alpha_{i\pi} \Phi(x_i, \pi)$.

\begin{algorithm}[tb]
  \caption{Direct Optimization of Ranking Measures}
  \label{alg:dorm}
  \begin{algorithmic}
    \STATE {\bfseries Input:} Document collections $D_i$, queries $q_i$,
    ranks $y_i$, sample size $m$, tolerance $\epsilon$
    \STATE Initialize $S_i = \emptyset$ for all $i$, and $w = 0$.
    \REPEAT
    \FOR{$i=1$ {\bfseries to} $m$} 
    \STATE $w = \sum_{i} \sum_{\pi \in S_i} \alpha_{i\pi} \Phi(x_i,
    \pi)$
    \STATE $\pi^* = \argmax_{\pi \in \Pi} \inner{w}{\Phi(x_i, \pi)} +
    \Delta(\pi, y_i)$
    \STATE $\xi = \max(0, \max_{\pi \in S_i}) \inner{w}{\Phi(x_i, \pi)} +
    \Delta(\pi, y_i)$
    \IF{$\inner{w}{\Phi(x_i, \pi^*)} + \Delta(\pi^*, y_i) > \xi + \epsilon$}
    \STATE Increase constraint set $S_i \leftarrow S_i \cup {\pi^*}$
    \STATE Optimize \eq{eq:dual} using only $\alpha_{i\pi}$ where $\pi \in S_i$.
    \ENDIF
    \ENDFOR
    \UNTIL{$S$ has not changed in this iteration}
  \end{algorithmic}
\end{algorithm}

\section{Extensions}

\subsection{Diversity Constraints}

Imagine the following scenario: when searching for 'Jordan', we will
find many relevant webpages containing information on this subject. They
will cover a large range of topics, such as a basketball player (Mike
Jordan), a country (the kingdom of Jordan), a river (in the Middle
East), a TV show (Crossing Jordan), a scientist (Mike Jordan), a city
(both in Minnesota and in Utah), and many more. Clearly, it is desirable
to provide the user with a diverse mix of references, rather than
exclusively many pages from the same site or domain or topic range. 

One way to achieve this goal is to include an interaction term between
the items to be ranked. This leads to optimization problems of the form
\begin{align}
  \mini_{\pi \in \Pi} \sum_{ijkl} \pi_{ij} \pi_{kl} c_{ij,kl}
\end{align}
where $c_{ij,kl}$ would encode the interaction between items. This is
clearly not desirable, since problems of the above type cannot be solved
in polynomial time. This would render the algorithm impractical for
swift ranking and retrieval purposes. 

However, we may take a more pedestrian approach, which will yield
equally good performance in practice, without incurring exponential
cost. This approach is heavily tailored towards ranking scores which
only take the top $n$ documents into account. We will require that among
the top $n$ retrieved documents no more than one of them may come from
the same source (e.g.\ topic, domain, subdomain, personal
homepage). Nonetheless, we would like to minimize the ranking scores
subject to this condition. 
Formally, we would like to find a matrix $\pi \in \cbr{0, 1}^{l \times
  n}$ such that
\begin{align}
  \sum_{i,j} \pi_{ij} a_i b(y)_j
\end{align}
is maximized, subject to the constraints
\begin{align}
  \sum_{i} \pi_{ij} & = 1 \text{ for all } j \in \cbr{1, n} \\
  \sum_{j} \sum_{i \in B_s} \pi_{ij} & \leq 1 \text{ for all } B_s.
\end{align}
Here the disjoint sets $B_s$ which form a partition of $\cbr{1, \ldots l}$,
correspond to subsets of documents (or webpages) which must not be
retrieved simultaneously. In the above example, for instance all
webpages retrieved from the domain
\emph{http://www.jordan.govoffice.com} would be lumped together into one
set $B_s$. Another set $B_{s'}$ would cover, e.g.\ all webpages from
\emph{http://www.cs.berkeley.edu/$\sim$jordan/}.

It is not difficult to see that during training, we need to solve an
optimization problem of the form
\begin{subequations}
  \label{eq:boringset}
\begin{align}
  \maxi_{\pi} \ & \tr \pi^\top C \\
  \label{eq:boringset-cons}
  \text{subject to } & \sum_{i} \pi_{ij} = 1 \text{ and } \sum_{j}
  \sum_{i \in B_s} \pi_{ij} \leq 1
  \text{ where } 
  \pi \in \cbr{0, 1}^{l \times n}. 
\end{align}
\end{subequations}
We will show below that the constraint matrix of \eq{eq:boringset} is
totally unimodular. This means that a linear programming relaxation of
the constraint set, i.e.\ the change from $\pi_{ij} \in \cbr{0, 1}$ to
$\pi_{ij} \in \sbr{0, 1}$ will leave the solution of the problem
unchanged. This can be seen as follows:

\begin{theorem}[\cite{HelTom56}]
  \label{th:heltom}
  An integer matrix $A$ with $A_{ij} \in \cbr{0, \pm 1}$ is totally
  unimodular if no more than two nonzero entries appear in any column,
  and if its rows can be partitioned into two sets such that:
  \begin{enumerate}
  \item If a column has two entries of the same sign, their rows are in
    different sets;
  \item If a column has two entries of different signs, their rows are
    in the same set. 
  \end{enumerate}
\end{theorem}

\begin{corollary}
  The linear programming relaxation of \eq{eq:boringset} has an
  integral solution.
\end{corollary}
\begin{proof}
  All we need to show is that in \eq{eq:boringset-cons} each term
  $\pi_{ij}$ only shows up exactly twice with coefficient $1$. This is
  clearly the case since $B_s$ is a partition of $\cbr{1, \ldots, l}$,
  which accounts for one occurrence, and the assignment constraints which
  account for the other occurrence. Hence Theorem~\ref{th:heltom}
  applies. 
\end{proof}
Note that we could extend this further by requiring that weighted
combinations over topics $\sum_{ij} \pi_{ij} w_{is} \leq 1$, where now
the weights $w_{is}$ may be non-integral and the domains where $w_{is}$
is nonzero might overlap. In this case, obviously the optimization
problem cannot be relaxed easily any more. However, it will still
provide useful results when used in combination with integer programming
codes, such as Bonmin \citep{Bonamietal05}. 

Finally, note that at test stage, it is very easy to take the
constraints \eq{eq:boringset-cons} into account: Simply pick the highest
ranking document from each set $B_s$ and use the latter to obtain an
overall ranking. 

\subsection{Ranking Matrix Factorization}

An obvious application of our framework is matrix factorization for
collaborative filtering. The work of
\cite{SreShr05,RenSre05,SreRenJaa05} suggests that regularized matrix
factorizations are a good tool for modeling collaborative filtering
applications. More to the point, Srebro and coworkers assume that they
are given a sparse matrix $X$ arising from collaborative filtering,
which they would like to factorize. 

More specifically, the entries $X_{ij}$ denote ratings by user $i$ on
document/movie/object $j$. The matrix $X \in \RR^{m \times n}$ is
assumed to be sparse, where zero entries correspond to (user,object)
pairs which have not been ranked yet. The goal is to find a pair of
matrices $U, V$ such that $U V^\top$ is close to $X$ for all nonzero
entries. Or more specifically, such that the entries $[U V^\top]_{ij}$
can be used to recommend additional objects. 

However, this may not be a desirable approach, since it is, for
instance, completely irrelevant how accurate our ratings are for
undesirable objects (small $X_{ij}$), as long as we are able to capture
the users preferences for desirable objects (large $X_{ij}$)
accurately. In other words, we want to model the user's \emph{likes} well,
rather than his \emph{dislikes}. In this sense, any indiscriminate
minimization, e.g.\ of a mean squared error, or a large margin error for
$X_{ij}$ is inappropriate. 

Instead, we propose to use a ranking score such as those proposed in
Section~\ref{sec:scoresloss} to evaluate an entire \emph{row} of $X$ at
a time. That is, we want to ensure that $X_{ij}$ is well reflected as a
whole in the estimates for \emph{all objects $j$ for a fixed user $i$}. This means
that we should be minimizing 
\begin{align}
  R_\emp[U, V, X] := \frac{1}{m} \sum_{i} \Delta([U V^\top]_{i\cdot}, X_{i\cdot})
\end{align}
where $\Delta$ is defined as in \eq{eq:loss} and it is understood that
it is evaluated over the nonzero terms of $X_{ij}$ only. This is a
highly nonconvex optimization problem. However, we can, again, find a
convex upper bound by the methods described in \eq{eq:primal}, yielding
a function $\tilde{R}_\emp[U,V,X]$. The technical details are
straightforward and therefore omitted. 

Note that by linearity this upper bound is convex in $U$ and $V$
respectively, whenever the other argument remains fixed. Moreover, note
that $\tilde{R}[U,V,X]$ decomposes into $m$ independent problems in
terms of the users $U_{i\cdot}$, whenever $V$ is fixed, whereas no such
decomposition holds in terms of $V$. 

In order to deal with overfitting, regularization of the matrices $U$
and $V$ is recommended. The trace norm $\nbr{U}_F^2 + \nbr{V}_F^2$ can
be shown to have desirable properties in terms of generalization
\citep{SreAloJaa05}. This suggests an iterative procedure for
collaborative filtering:
\begin{itemize}
\item For fixed $V$ solve $m$ independent optimization problems in terms
  of $U_{i\cdot}$, using the Frobenius norm regularization on $U$. 
\item For fixed $U$ solve one large-scale convex optimization problem in
  terms of $V$. 
\end{itemize}
Since the number of users is typically considerably higher than the
number of objects, it is possible to deal with the optimization problem
in $V$ efficiently. Details are subject to future research. 

\section{Experiments}
\label{sec:experiments}

We address a number of questions: a) is learning
needed for good ranking performance, b) how does DORM (our algorithm)
perform with respect to other algorithms on a number of datasets of
different size and truncation level of the performance criteria, c) how
fast is DORM when compared to similar large margin ranking algorithms,
and d) how important is the choice of $c$ for good performance? 



%

\subsection{Datasets and Experimental Protocol}


\begin{description}
\item[UCI] 
We choose PageBlock, PenDigits, OptDigits, and Covertype
  from the UCI repository mainly to increase the number of different
  datasets on which we may compare DORM to other existing approaches.
  Since they are not primary ranking data, we will discuss the outcomes
  only briefly for illustrative purposes. For PageBlock, PenDigits and
  OptDigits we sample 50 queries and 100 documents for each query. For
  Covertypes we sample 500 queries and 100 documents. 
\item[Web Seach]
  Our web search dataset (courtesy of Chris Burges at Microsoft
  Research) consists of 1000 queries each for training, validation and
  testing. They are provided and selected from a larger pool of training
  data used for a search engine.
  Figure~\ref{fig:stats} shows a histogram of the number of documents
  per query (the median is approximately 50). Documents are ranked
  according to five levels of relevance (1:Bad, 2:Fair, 3:Good,
  4:Excellent, 5:Perfect). Unlabelled documents are treated as Bad. The ratio between
  the five categories (1 to 5) is approximately 75:17:15:2:1. The length
  of the feature vectors is 367 (i.e.\ we are using BM25).
  We evaluate our algorithm with respect to three goals: performance in
  terms of NDCG@$n$, MRR and WTA performance. 
\item[EachMovie]
  This collaborative filtering dataset consists of 2811983 ratings by
  72916 users on 1628 movies. To prove the point that our improvement is
  \emph{not} due to an improved choice of a kernel but rather in the
  improved choice of a \emph{loss function}, we follow the experimental
  setup, choice of kernels, and pre-processing of \citep{BasHof04} and
  compare performance using ERU. We also use the experimental setup of
  \citep{YuYuTreKri06} and compare performance using NDCG and NDCG@10. In
  both cases we are able to improve the results considerably. The
  datasets for both experiments are as provided by Thomas Hofmann at
  Google Research, and Shipeng Yu at Siemens Research respectively.

\item[Protocol] 
Since WebSearch provided  a validation set, we used the latter
for model selection. Otherwise, 10-fold cross validation was used to
adjust the regularization constant $C$. We used linear kernels
throughout, except for the EachMovie datasets, where we followed the
protocols of \citep{BasHof04} and \citep{YuYuTreKri06}. 
This was done to show that the performance improvement we observe
is due to our choice of a better loss function rather than the function
class. Note that NDCG, MRR were rescaled from  $[0,1]$ to $[0,100]$ 
for better visualization.
\end{description}

\subsection{UCI Datasets}

Since the UCI data does not come in the form of multiple queries, we
permute the datasets and randomly subsample documents for each query.

We compare DORM against SVM classification, Precision@n and ROCArea as
reported by \citep{Joachims05} and use NDCG@10 as a performance measure.
Table \ref{table:toy-result} shows that our algorithm significantly
outperforms other methods. $c$ was chosen to be $c_i = 1/(i+1)$ in DORM.
The experiment results show how hard it is to optimize NDCG,
traditional methods perform poorly when we use NDCG as the performance
metric while direct optimization of the correct criterion works very well.
\begin{table}[bt]
  \centering
  \begin{tabular}{|l||r|r|r|r|} \hline
    Dataset & ROCArea & SVM & Prec@10 & DORM \\ \hline
    PageBlocks & $35.9 \pm 7$ & $46.5 \pm 7$& $44.0 \pm 7$& ${\bf 63.7
    \pm 6}$ \\
    PenDigits & $26.2 \pm 8$& $41.5 \pm 4$& $15.6 \pm 3$& ${\bf
    85.2 \pm 3}$ \\
    OptDigits & $26.0 \pm 9$& $26.1 \pm 3$ & $26.2 \pm  3$& ${\bf 76.1 \pm 6}$ \\
    Covertypes & $47.0 \pm 2$& $48.5 \pm 2$& $42.1 \pm 1$& ${\bf
    58.8 \pm 2}$ \\ \hline
  \end{tabular}
  \caption{Performance on UCI data. Bold indicates
  a high significance of $p < 0.0001$ by paired {\it t}-test. \label{table:toy-result}}
\end{table}

\subsection{NDCG for Web Search}
\label{sec:web-search}

\begin{figure}
  \centering
  \myfig{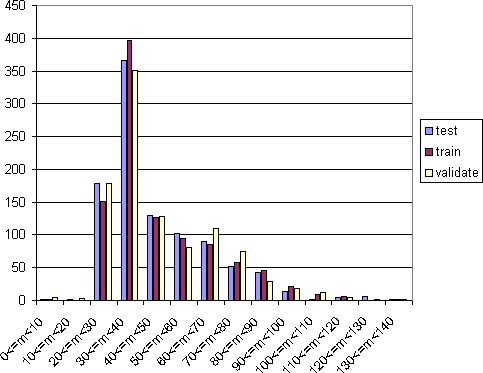}
  \caption{Number of documents per query.\label{fig:stats}}
\end{figure}
\begin{figure}
  \centering
  \myfig{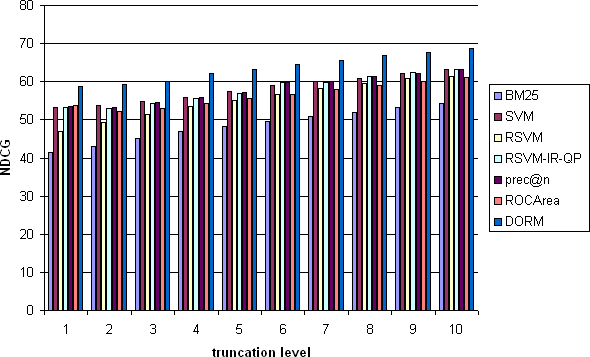}
  \caption{NDCG@n scores on the web search dataset at different
    truncation levels $n$. \label{fig:msn1}}
\end{figure}
\begin{figure}
  \centering
  \myfig{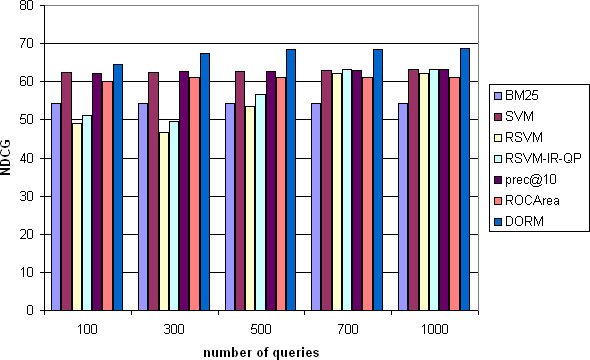}
  \caption{NDCG@10 scores on the web search dataset for different sample
    sizes. Note that the performance of NDCG on 100 observations is
    larger than of any other competing method for 1000 observations.
    \label{fig:msn2}}
\end{figure}

We compare DORM to a range of kernel methods for ranking: multiclass SVM
classifiers, SVM for ordinal regression (RSVM) \citep{Joachims02b,
  HerGraObe00}, SVM for information retrieval (SVM-IR-QP)
\citep{CaoXuLiuLiHuaHon06}, Precision@n, ROCArea \citep{Joachims05} and
DORM. We use NDCG to assess the performance of the ranking algorithms.
BM25 \citep{Robertsonetal94, RobHul00} is used as a baseline (it also constitutes the
feature vector for the other algorithms).

In the first experiment, we use the full training set (parameter
selection on the validation set) to train our model and report the
performance on the designated test set. We report the average prediction
results for NDCG with various truncation levels ranging from 1 to 10 in
Figure \ref{fig:msn1}. Note that DORM consistently outperforms other
methods for ranking by $2-3\%$. We chose $c_i = \rbr{i+1}^{-1/2}$.

In a second experiment (same choice of $c_i$), we investigate the
effects of increasing training set size using the first $m=$100, 300, 500,
700 or all 1000 queries for training. Using the same experimental
protocol as above we report the average prediction results for NDCG@10
in Figure \ref{fig:msn2}. Our results confirm that the same gain for
NDCG@10 can be achieved if we increase the training set size. Note that
for many methods, doubling the sample size increases NDCG by less than
$1\%$ (Figure \ref{fig:msn2}). DORM achieves the same gain
without the need to double $m$. In fact, DORM using 100
queries for training beats all other methods using 1000 queries!

Since expert-judged datasets can be very expensive, DORM is more
cost-effective.  This confirms that learning to rank is beneficial, as
using a combination of many features for ranking algorithms could result
in better performance than using one or some several features
individually, such as BM25.

The choice of $c$ critically determines the feature map. At the moment,
we have little theoretical guidance with regard to this matter, hence we
investigated the effect of choosing schemes of $c$ experimentally.
Clearly $c$ needs to be a monotonically decreasing function. We chose
$c_i = (i+1)^{-d}$ for $d \in \cbr{\frac{1}{4},\frac{1}{3}, \frac{1}{2}, 1, 2, 3}$
and $c_i = 1/\log(i+2)$ and $c_i = 1/\log \log(i+2)$.

We found experimentally that the differences between the various schemes
are not as dramatic as the improvement obtained by using DORM instead of
other algorithms. To summarize the results we show the difference in
performance in Figure~\ref{fig:difference} for NDCG@10. Note that the
difference in terms of NDCG accuracy resulted by taking different $c$
will decrease when the sample size increases. The rate of convergence
is suspected to be $1/\sqrt{m}$. An possible
interpretation is that the
choice of $c$ can be considered prior knowledge. Thus with 
increasing sample size, we will need to rely less on this prior
knowledge and a reasonable choice of $c$ will suffice. 


\begin{figure}[tb]
   \centering
   \includegraphics[width=\columnwidth, height=0.6\columnwidth]{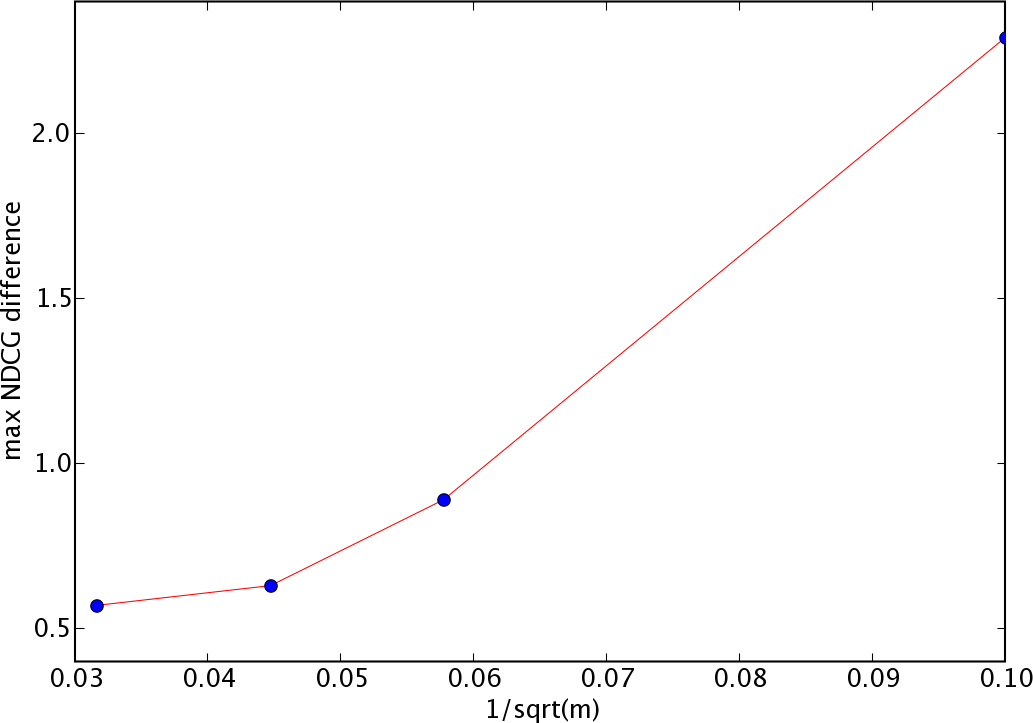}
   \caption{
     Maximum difference in NDCG@10 for different choices of $c$
     with respect to $1/\sqrt{m}$. \label{fig:difference}}
\end{figure}

\subsection{MRR and WTA for Web Search}

\paragraph{MRR}

DORM is effective not only for NDCG but also for other
performance measures. We compare it using the Mean Reciprocal Rank (MRR)
and Winner Takes All (WTA) scores on the same dataset. For comparison we
use Precision@n (with $n = 3$, since this yielded the best results
experimentally), DORM minimizing NDCG (which in this case is the
incorrect criterion), and the previous methods. 

As before, we use the validation set to adjust the regularization
parameter $C$. We picked $c_i = 1/(i+1)$ (the influence of the
particular choice of $c$ was rather minor, as demonstrated in
Figure~\ref{fig:difference}). The average results for varying sample
size (from 100 to 1000) are reported in Figure \ref{fig:msn2-mrr}. 

It can be seen from Figure~\ref{fig:msn2-mrr} that DORM for MRR outperforms all
other models including DORM for NDCG. This is not surprising, since DORM
for MRR optimizes MRR directly while other methods do not. DORM for MRR
beats other methods by $1\%$ to $2\%$ which is quite significant given
that if we double the dataset size, the gain is only around $1\%$. The
fact that the gains in optimizing MRR are less than when optimizing MRR
is probably due to the fact that MRR is less structured than NDCG.

\paragraph{WTA}

This is the least structured of all scores, as it only takes the top
ranking document retrieved into account. This means that for a ranking
dataset where $5$ different degrees of relevance are available, only the
top scoring ones are chosen.  This transformation discards a great deal
of information in the labels (i.e. the gradations among the
lower-scoring documents), which leads to the suspicion that minimizing a
related cost function taking all levels into account should perform
better. 
It turns out, experimentally, that indeed a
direct optimization of the WTA loss function leads to bad
performance. In order to amend this issue, we decided to minimize a
modified NDCG score instead of the straight WTA score. This improved the
performance significantly.

The truncation level in the NDCG@n score should be closer to 1 rather
than 10. We devise a heuristic to find the truncation level:
for queries that have more than 2 level of preference at the top 3
items, the truncation is 3; for queries that have only one level of
preference at the top 3 items, the truncation level is after position
where the next level of preference appears in the ranked list, as for
documents with a large number of top ranking documents we want to
include at least one lower-ranking document in the list. We call
the method mWTA (modified DORM for WTA). 

Experiments with different decay terms for $c$ indicated that $c_i =
1/\sqrt{i+1}$ yields best results. We compare the new method with various
methods using model selection on the validation set and report the total
number of correct predictions in Figure \ref{fig:msn3-wta}. RSVM
and RSVM-IR-QP perform poorly on this task and we omit their results due
to space constraints. While direct optimization of WTA is
unsatisfactory, mWTA outperforms other methods significantly.

\begin{figure}
  \centering
  \myfig{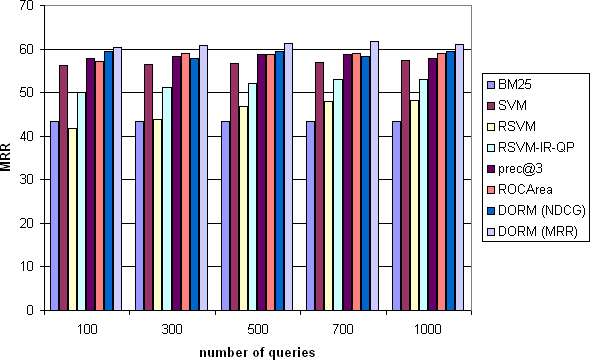}
  \caption{MRR scores on the web search dataset for different sample
    sizes. We compare eight methods (including DORM), using $c_i =
    1/(i+1)$. \label{fig:msn2-mrr}}
\end{figure}
\begin{figure}
  \centering
  \myfig{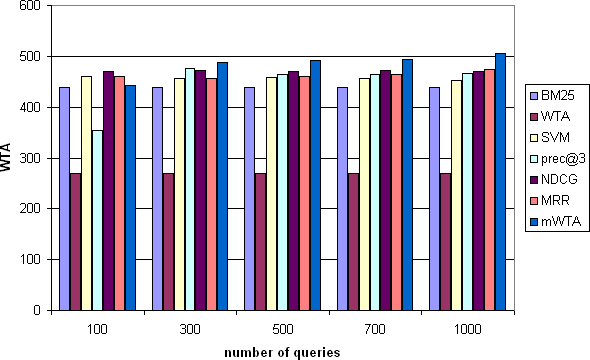}
  \caption{WTA scores (``I'm feeling lucky'') on the web search dataset
    for different sample sizes. DORM (mWTA) minimizing an adaptive
    version of NDCG outperforms straight WTA minimization, as it makes
    better use of the label information.  
    \label{fig:msn3-wta}}
\end{figure}
\begin{figure}
  \centering
   \includegraphics[width=0.7\textwidth, height=0.5\textwidth]{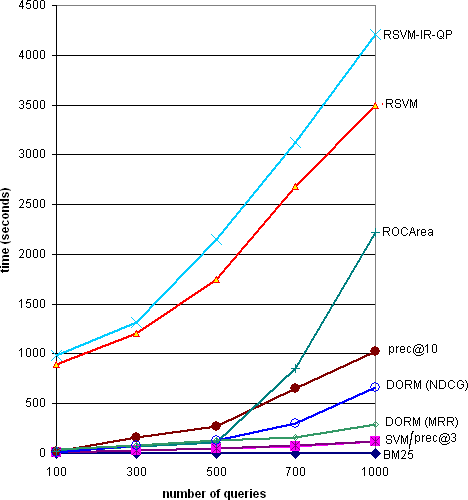}
   \caption{Runtime of DORM (NDCG and MRR
   optimization) vs.\ other algorithms, ignoring file IO. \label{fig:time}}
\end{figure}

\subsection{Runtime Performance}

One might suspect that our formulation is potentially slow, as the
Hungarian marriage algorithm takes cubic time $O(l^3)$. However, each
such optimization problem is relatively small (on average in the order
of 50 documents per query), which means that the overall computational
time is well controlled. 

For practical results we carried out experiments to measure the time for
training \emph{plus} cross validation. We used a modified version of
SVMStruct \citep{TsoJoaHofAlt05}. The algorithms are all written in $C$
and the code was run on a Pentium 4 3.2GHz workstation with 1GB RAM,
running Linux and using GCC 3.3.5. As can be seen in
Figure~\ref{fig:time}, DORM outruns most other methods, except for
multiclass SVM, and BM25 (which does not require training). Note that
ordinal regression algorithms are significantly slower than DORM, as
they need to deal with a huge number of simple inequalities rather than
a smaller number of more meaningful ones. 

\begin{table}
  \centering
  \begin{tabular}{|l||r|r|r|} \hline
    \multicolumn{4}{|c|}{NDCG@10 optimization} \\ \hline
    Method & \# SVs & \# Iter & \% in QP \\ \hline
    Precision@10 & $1037$ & $44$ & $89.02$\\ \hline
    ROCArea & $997$ & $12$& $19.64$\\ \hline
    DORM (NDCG)& $561$ & $22$ & $28.76$\\ \hline
    \multicolumn{4}{c}{ } \\ \hline
    \multicolumn{4}{|c|}{MRR optimization} \\ \hline
    Method & \# SVs & \# Iter & \% in QP \\ \hline
    Precision@3 & $1000$ & $11$ & $46.80$\\ \hline
    ROCArea & $997$ & $12$& $19.64$\\ \hline
    DORM (NDCG) & $520$ & $23$ & $61.66$\\ \hline
    DORM (MRR) & $550$ & $17$ & $1.76$\\ \hline
  \end{tabular}
  \medskip
  \caption{Number of support vectors, iterations in column generation,
    and time spent in the Quadratic Programming loop for various SVM
    style optimization algorithms. Top: optimization for NDCG@10, using
    $c_i = 1/\sqrt{i+1}$. Bottom: optimization for MRR and 
    NDCG, using $c_i = 1/(i+1)$. Corresponding methods have different
    numbers due to the different choices of truncation level (Precision@n)
    and $c$ (DORM NDCG). 
    \label{tab:svsqp}}
\end{table}

\begin{table}
  \centering
  \begin{tabular}{|c||c|c|c|} \hline
    \multicolumn{1}{|c||}{Experiment}
    &\multicolumn{1}{c|}{PRank}
    &\multicolumn{1}{c|}{JRank}
    &\multicolumn{1}{c|}{DORM}\\ \hline 
    1 & $70.8$ & $75.3$ & $76.5 \pm 0.43$\\ \hline
    2 & $73.4$ & $76.2$& $76.7 \pm 0.32$\\ \hline
  \end{tabular}
  \medskip
  \caption{Expected rank utility scores for three
  methods. Results averaged over 100 trials with 100 training users,
  2000 input users and 800 training items.   \label{result-movie}}

  \begin{tabular}{|c|c|c|c|} \hline
    \multicolumn{1}{|c|}{N}
    &\multicolumn{1}{c|}{Method}
    &\multicolumn{1}{c|}{NDCG}
    &\multicolumn{1}{c|}{NDCG@10}\\\hline
    10 & GPR & $83.41 \pm 0.22$ & $45.58 \pm 1.51$\\
       & CGPR & $86.39 \pm 0.24$ & $57.34 \pm 1.44$\\
       & GPOR & $80.59 \pm 0.03$ & $36.92 \pm 0.25$\\
       & CGPOR &$80.83 \pm 0.11$ & $37.89 \pm 1.05$\\
       & MMMF & $84.34 \pm 0.48$ & $47.46 \pm 3.42$\\
       & DORM & ${\bf 87.17 \pm 0.24}$ & ${\bf 61.75 \pm 1.83}$\\\hline
       &      & $p < 0.0001$& $p < 0.0001$\\\hline
   20  & GPR & $ 84.12 \pm 0.15$ & $48.49 \pm 0.66$\\
       & CGPR & $86.98 \pm 0.16$ & $59.89 \pm 1.18$\\
       & GPOR & $80.48 \pm 0.05$ & $36.78 \pm 0.30$\\
       & CGPOR &$80.78 \pm 0.13$ & $37.81 \pm 0.56$\\
       & MMMF & $84.85 \pm 0.28$ & $47.86 \pm 1.39$\\
       &DORM & ${\bf 87.63 \pm 0.37}$ & ${\bf 62.82 \pm 1.9}$\\\hline
       &     &$p < 0.0001$&$ p = 0.0006$\\\hline
   50  &GPR & $85.15 \pm 0.23$ & $53.75 \pm 0.89$\\
       &CGPR &$ 87.82 \pm 0.21$&$63.41 \pm 1.14$\\
       &GPOR &$80.10 \pm 0.04$& $36.63 \pm 0.24$\\
       &CGPOR &$80.45 \pm 0.06$ & $37.74 \pm 0.41$\\
       &MMMF &$86.13 \pm 0.38$& $54.78 \pm 2.11$\\
       &DORM &$87.84 \pm 0.32$& ${\bf 65.05 \pm 1.27}$\\\hline
       &     &$p = 0.8706$&$p = 0.006$\\\hline
  \end{tabular}
  \medskip
  \caption{NDCG optimization on EachMovie
  dataset. Comparison between 6 methods using unpaired {\it t}-test
  with values of $p$ shown (best score vs. second best score).
  \label{table:eachmovie-ndcg}}
\end{table}


Table~\ref{tab:svsqp} has a comparison of the number of support
vectors (the more the slower) number of column generation iterations
(the more the slower), percent of time spent in the QP solver. 
DORM is faster than other algorithms since it has a 
sparser solution. In terms of number of iterations a percent of time in
the QP solver, DORM is a well balanced solution between Precision@n and
ROCArea.
Results are similar when optimizing MRR 
(this is reported in the
bottom of Table~\ref{tab:svsqp}). 
%
Note that since all models use linear functions, prediction times is
less than 0.5s for 1000 queries.

\subsection{EachMovie and Collaborative Filtering}
\label{sec:eachmovie}

\paragraph{ERU}

Past published results on collaborative filtering use
Expected Rank Utility (ERU), NDCG and NDCG@10 as reference scores. In
order to show that the improvement in performance is truly due to a
better loss function rather than a different kernel we use the same
kernels and experimental protocol as proposed by \citep{BasHof04} using the same parameter
combinations in the context of ERU. Table \ref{result-movie} shows the
merit of DORM: it outperforms JRank \citep{BasHof04} and PRank
\citep{CraSin02}. In experiment 1 we used user features in combination
with item correlations. In experiment 2 we used item features in
combination with user ratings. In both cases, results are averaged over
100 trials with 100 training users, 2000 input users and 800 training
items. 

Having used a kernel which is optimal for JRank we expect that
optimizing the kernel further would lead to better results, as there is
no reason to assume that the model class optimal for JRank would be the
best choice for DORM, too.

\paragraph{NDCG}

In a second experiment, we mimicked the experimental
protocol of \citep{YuYuTreKri06} on EachMovie. Here, we treat each movie
as a document and each user as a query. After filtering out all the
unpopular documents and queries (as in \citep{YuYuTreKri06}) we have 1075
documents and 100 users. 

For each user, we randomly select 10, 20 and 50 labeled items for
training and perform prediction on the rest. The process is repeated 10
times independently. The methods for comparison are the standard
Gaussian process regression (GPR) \citep{RasWil06}, Gaussian Process
ordinal regression (GPOR) \citep{ChuGha05}, and their collaborative
extensions (CPR, CGPOR) \citep{YuYuTreKri06}, MMMF \citep{RenSre05} and
DORM (for NDCG). The figures for the first 5 methods are extracted from
\citep{YuYuTreKri06} and scaled by 100 to fit our convention of showing
NDCG results.  We perform an unpaired {\it t}-test for significance
(see Table \ref{table:eachmovie-ndcg}).

Note that there is no cross validation or model selection involved in
\citep{YuYuTreKri06}. Thus to be fair, we fix the following parameters
$c_i = {(i+1)^{-0.25}}$ (performing slightly better in this dataset), $C
= 0.01$ and $n = 10$ (the truncation level of NDCG).
The results show that DORM performs very well for predicting the
ranking for new items, especially when the number of labeled items is
small.

\section{Summary and Discussion}

In this paper we proposed a general scheme to deal with a large range of
criteria commonly used in the context of web page ranking and
collaborative filtering. Unlike previous work, which mainly focuses on
pairwise comparisons we aim to minimize the multivariate performance
measures (or rather a convex upper bound on them) directly. This has
both computational savings, leading to a faster algorithm and
practical ones, leading to better performance. In a way, our work
follows the mantra of \citep{Vapnik82} of estimating \emph{directly} the
desired quantities rather than optimizing a surrogate function. 
There are clear extensions of the current work:
\begin{itemize}
\item The key point of our paper was to construct a well-designed loss
  function for optimization. In this form it is completely generic and
  can be used as a drop-in replacement in many settings. We completely
  ignored language models \citep{PonCro98} to parse the queries in any
  sophisticated fashion.
\item Although the content of the paper is directed towards ranking, the
  method can be generalized for optimizing many other complicated
  multivariate loss functions. 
\item We could use our method directly for information retrieval tasks
  or authorship identification queries. In the latter case, the query
  $q_i$ would consist of a collection of documents written by one
  author.
\item We may add personalization to queries. This is no major problem,
   as we can simply add some personal data $u_i$ to $\phi(q_i, d_i, u_i)$
   and obtained personalized ranking.
\item Online algorithms along the lines of \citep{ShaSin06} can easily be
  accommodated to deal with massive datasets. 
\item The present algorithm can be extended to learn  matching
  problems on graphs. This is achieved by extending the linear assignment problem
  to a quadratic one. The price one needs to pay in this case is that
  the Hungarian Marriage algorithm is no longer feasible, as the
  optimization problem itself is NP hard. 
\end{itemize}
Note that the choice of a Hilbert space for the scoring functions is
done for reasons of convenience. If the applications demand Neural
Networks or similar (harder to deal with) function classes instead of
kernels, we can still apply the large margin formulation. That said, we
find that the kernel approach is well suited to the problem.

\paragraph{Acknowledgments:} We are indebted to Thomas Hofmann, Chris Burges,
and Shipeng Yu for providing us with their datasets for the purpose of
ranking. This was invaluable in obtaining results comparable with their own
publications (as reported in the experiments).
We thank Yasemin Altun, Chris Burges, Tiberio Caetano, David Hawking,
Bhaskar Mehta, Bob Williamson, and
Volker Tresp for helpful discussions. Part of this work was carried out
while Quoc Le was with NICTA. National ICT Australia is funded through
the Australian Government's \emph{Backing Australia's Ability}
initiative, in part through the Australian Research Council. This work
was supported by the Pascal Network.

\bibliography{../../../bibfile}

\end{document}